\documentclass[english,runningheads,10pt]{llncs}

\usepackage{tabularx,booktabs,multirow,delarray,array}
\usepackage{graphicx,amssymb,amsmath,amssymb}
\usepackage{enumerate}
\usepackage[ruled,vlined,linesnumbered]{algorithm2e}
\usepackage{fullpage}
\usepackage{latexsym}

\newenvironment{proof}{\par\noindent{\bf Proof:}}{\mbox{}\hfill$\qed$\\}

\newcommand{\ignore}[1]{ }

\newcounter{rem}
\def\qed{\hbox{\rlap{$\sqcap$}$\sqcup$}}

\begin{document}

\title{Vertex Fault-Tolerant Geometric Spanners for Weighted Points}

\author{
Sukanya Bhattacharjee\inst{1}
\and
R. Inkulu\inst{1}
\thanks{A preliminary version appeared in the Proceedings of the $25^{th}$ International Computing and Combinatorics Conference (COCOON 2019).}
}

\institute{
Department of Computer Science \& Engineering\\
IIT Guwahati, India\\
\email{\{bsukanya,rinkulu\}@iitg.ac.in}
}

\maketitle

\pagenumbering{arabic}
\setcounter{page}{1}

\begin{abstract}
Given a set $S$ of $n$ points, a weight function $w$ to associate a non-negative weight to each point in $S$, a positive integer $k \ge 1$, and a real number $\epsilon > 0$, we present algorithms for computing a spanner network $G(S, E)$ for the metric space $(S, d_w)$ induced by the weighted points in $S$.
The weighted distance function $d_w$ on the set $S$ of points is defined as follows:
for any $p, q \in S$, $d_w(p, q)$ is equal to $w(p) + d_\pi(p, q) + w(q)$ if $p \ne q$, otherwise, $d_w(p, q)$ is $0$.
Here, $d_\pi(p, q)$ is the Euclidean distance between $p$ and $q$ if points in $S$ are in $\mathbb{R}^d$, otherwise, it is the geodesic (Euclidean) distance between $p$ and $q$.
The following are our results:
(1) When the weighted points in $S$ are located in $\mathbb{R}^d$, we compute a $k$-vertex fault-tolerant $(4+\epsilon)$-spanner network of size $O(k n)$.
(2) When the weighted points in $S$ are located in the relative interior of the free space of a polygonal domain $\cal P$, we detail an algorithm to compute a $k$-vertex fault-tolerant $(4+\epsilon)$-spanner network with $O(\frac{kn\sqrt{h+1}}{\epsilon^2} \lg{n})$ edges.
Here, $h$ is the number of simple polygonal holes in $\cal P$.
(3) When the weighted points in $S$ are located on a polyhedral terrain $\cal T$, we propose an algorithm to compute a $k$-vertex fault-tolerant $(4+\epsilon)$-spanner network, and the number of edges in this network is $O(\frac{kn}{\epsilon^2} \lg{n})$.
\end{abstract}


\section{Introduction}
\label{sect:intro}

In designing geometric networks on a given set of points in a metric space, it is desirable for the network to have short paths between any pair of nodes while being sparse with respect to the number of edges.
For a set $S$ of $n$ points in a metric space $\cal M$, a network on $S$ is an undirected graph $G$ with vertex set $S$ and an edge set $E$, where every edge of $G$ associated with a weight. 
The distance in $G$ between any two vertices $p$ and $q$ of $G$, denoted by $d_G(p, q)$, is the length of a shortest (that is, a minimum length) path between $p$ and $q$ in $G$.
For a real number $t \ge 1$, the graph $G$ is called a {\it $t$-spanner} of points in $S$ if for every two points $p, q \in S$, $d_G(p, q)$ is at most $t$ times the distance between $p$ and $q$ in $\cal M$.
The smallest $t$ for which $G$ is a $t$-spanner of points in $S$ is called the {\it stretch factor} of $G$, and the number of edges of $G$ is called its size.
Given a set $S$ of points in plane, each associated with a non-negative weight, and a positive integer $k$, in this paper, we study computing an edge-weighted geometric graph $G$ that is a {\it vertex fault-tolerant $t$-spanner} for the metric space induced by the weighted points in $S$;
that is, for any set $S' \subseteq S$ of cardinality at most $k$, the graph $G \setminus S'$ is a $t$-spanner for the metric space induced by the weighted points in $S-S'$.

\subsection*{Previous Work}

Peleg and Sch\"{a}ffer~\cite{journals/jgt/PelegSchaffer89} and Chew~\cite{journals/jcss/Chew89} introduced spanner networks.
Alth\"{o}fer et~al.~\cite{journals/dcg/AlthoferDDJS93} studied sparse spanners on edge-weighted graphs with edge weights obeying the triangle-inequality.
The text by Narasimhan and Smid \cite{books/compgeom/narsmid2007}, and the handbook chapters by Eppstein \cite{hb/cg/Epp99} and Gudmundsson and Knauer~\cite{hb/apprxheu/GudKnau07} detail various results on Euclidean spanners, including a $(1+\epsilon)$-spanner for the set $S$ of $n$ points in $\mathbb{R}^d$ that has $O(\frac{n}{\epsilon^{d-1}})$ edges, for any $\epsilon > 0$.

Apart from the small number of edges, spanners with additional properties such as small weight, bounded degree, small diameter, planar network, etc., were also considered.
The significant results in optimizing these parameters in geometric spanner network design include 
spanners of low degree \cite{journals/cgta/ABCGHSV08,conf/cccg/CarmiChai10,journals/cgta/BCCCKL13}, 
spanners of low weight \cite{journals/algorithmica/BCFMS10,journals/ijcga/DasNara97,journals/siamjc/GudmudLevcoNar02}, 
spanners of low diameter \cite{conf/focs/AryaMS94,journals/cgta/AryaMountSmid99}, 
planar spanners \cite{conf/esa/ArikatiCCDSZ96,journals/jcss/Chew89,conf/optalgo/DasJoseph89,conf/wads/KeilGutwin89}, 
spanners of low chromatic number \cite{journals/cgta/BCCMSZ09}, 
fault-tolerant spanners \cite{journals/dcg/ABFG09,journals/dcg/CzumajZ04,journals/algorithmica/LevcopoulosNS02,conf/wads/Lukovszki99,journals/corr/KapLi13,conf/stoc/Solomon14,conf/caldam/Inkulu19b}, 
low power spanners \cite{journals/wireless/Karim11,conf/infocom/SegalShpu10,journals/jco/WangLi06}, 
kinetic spanners \cite{journals/dcg/AbamBerg11,journals/cgta/ABG10}, 
angle-constrained spanners \cite{journals/jocg/CarmiSmid12}, 
and combinations of these \cite{conf/stoc/AryaDMSS95,journals/algorithmica/AryaS97,journals/algorithmica/BFRV18,journals/algorithmica/BoseGudSmid05,journals/ijcga/BoseSmidXu09,journals/corr/CarmiChait10}.
For the case of spanners in a metric space with bounded doubling metric, a few results are given in \cite{conf/stoc/Talwar04}.

As observed in Abam et~al.,~\cite{journals/algorithmica/AbamBFGS11}, the cost of traversing a path in a network is not only determined by the lengths of edges along the path, but also by the delays occurring at the vertices on the path.
The result in \cite{journals/algorithmica/AbamBFGS11} models these delays by associating non-negative weights to points. 
Let $S$ be a set of $n$ points in $\mathbb{R}^d$. 
For every $p \in S$, let $w(p)$ be the non-negative weight associated to $p$. 
The following weighted distance function $d_w$ on $S$ defining the metric space $(S, d_w)$ is considered by Abam et~al. in \cite{journals/algorithmica/AbamBFGS11}, and by Bhattacharjee and Inkulu in \cite{conf/caldam/Inkulu19b}:
for any $p, q \in S$, $d_w(p, q)$ is equal to $w(p) + |pq| + w(q)$ if $p \ne q$; otherwise, $d_w(p, q)$ is equal to $0$.

Recently, Abam et~al.~\cite{conf/soda/AbamBS17} showed that there exists a $(2 + \epsilon)$-spanner with a linear number of edges for the metric space $(S, d_w)$ that has bounded doubling dimension.
And, \cite{journals/algorithmica/AbamBFGS11} gives a lower bound on the stretch factor, showing that $(2+\epsilon)$ stretch is nearly optimal.
Bose et~al. \cite{conf/swat/BoseCC08} studied the problem of computing a spanner for a set of weighted points in ${\mathbb R}^2$, while defining the distance $d_w$ between any two distinct points $p, q \in S$ as $d(p,q) - w(p) - w(q)$. 
Under the assumption that the distance between any pair of points is non-negative, they showed the existence of a $(1 + \epsilon)$-spanner with $O(\frac{n}{\epsilon})$ edges. 

A set of $h \ge 0$ disjoint simple polygonal holes (obstacles) contained in a simple polygon $P$ is the {\it polygonal domain} ${\cal P}$.
Note that when $h = 0$, the polygonal domain $\cal P$ is essentially a simple polygon.
The free space $\cal D$ of a polygonal domain ${\cal P}$ is defined as the closure of $P$ excluding the union of the interior of polygons contained in $P$.
Note that the free space of a simple polygon is its closure.
A shortest path between any two points $p$ and $q$ is a path in $\cal D$ whose length (in Euclidean metric) is less than or equal to the length of any path between $p$ and $q$ located in $\cal D$.
The distance along any shortest path between $p$ and $q$ is denoted by $d_\pi(p, q)$.
If the line segment joining $p$ and $q$ is in $\cal D$, then the Euclidean distance between $p$ and $q$ is denoted by $d(p, q)$, i.e., in this case, $d_\pi(p, q)$ is equal to $d(p, q)$.

Given a set $S$ of $n$ points in the free space $\cal D$ of $\cal P$, computing a geodesic spanner of $S$ is considered in Abam et~al.~\cite{conf/compgeom/AbamAHA15}.
The result in \cite{conf/compgeom/AbamAHA15} showed that for the metric space $(S, d_\pi)$, for any constant $\epsilon > 0$, there exists a $(5+\epsilon)$-spanner of size $O(\sqrt{h}n (\lg n)^{2})$. 
Further, when the input points are located in a simple polygon, for any constant $\epsilon > 0$, \cite{conf/compgeom/AbamAHA15} devised an algorithm to compute a $(\sqrt{10}+\epsilon)$-spanner with $O(n (\lg n)^{2})$ edges.

A polyhedral terrain $\cal T$ is the graph of a piecewise linear function $f: D \rightarrow {\mathbb R}^3$, where $D$ is a convex polygonal region in the plane.
Given a set $S$ of $n$ points on a polyhedral terrain $\mathcal{T}$, the geodesic distance between any two points $p, q \in S$ is the distance along any shortest path on the terrain between $p$ and $q$.
The spanner for points on a terrain with respect to geodesic distance on terrain is a geodesic spanner.
Unlike metric space induced by Euclidean distance among points in ${\mathbb R}^d$, the metric space induced by points on a terrain does not have a bounded doubling dimension.
Hence, geometric spanners for points on polyhedral terrains has unique characteristics and interesting to study.
The algorithm in \cite{conf/soda/AbamBS17} proved that for a set of unweighted points on any polyhedral terrain, for any constant $\epsilon > 0$, there exists a $(2 + \epsilon)$-geodesic spanner with $O(n \lg n)$ edges.

A graph $G(S, E)$ is a {\it $k$-vertex fault-tolerant $t$-spanner}, denoted by $(k, t)$-VFTS, for a set $S$ of $n$ points in $\mathbb{R}^d$ whenever for any subset $S'$ of $S$ with size at most $k$, the graph $G \setminus S'$ is a $t$-spanner for the points in $S \setminus S'$.
The algorithms given in Levcopoulos et~al., \cite{journals/algorithmica/LevcopoulosNS02}, Lukovszki \cite{conf/wads/Lukovszki99}, and Czumaj and Zhao \cite{journals/dcg/CzumajZ04} compute a $(k, t)$-VFTS for the set $S$ of points in $\mathbb{R}^d$.
These algorithms are also presented in \cite{books/compgeom/narsmid2007}.
Levcopoulos et~al.~\cite{journals/algorithmica/LevcopoulosNS02} devised an algorithm to compute a $(k, t)$-VFTS of size $O(\frac{n}{(t-1)^{(2d -1)(k+1)}})$ in $O(\frac{n \lg{n}}{(t-1)^{4d -1}} + \frac{n}{(t-1)^{(2d -1)(k+1)}})$ time,
and another algorithm to compute a $(k, t)$-VFTS with $O(k^{2}n)$ edges in $O(\frac{k n \lg n}{(t-1)^{d}})$ time. 
The result in \cite{conf/wads/Lukovszki99} gives an algorithm to compute a $(k, t)$-VFTS of size $O(\frac{k n}{(t-1)^{d-1}})$ in $O(\frac{1}{(t-1)^{d}}(n \lg^{d-1} n \lg k + k n \lg \lg n))$ time. 
The algorithm in \cite{journals/dcg/CzumajZ04} computes a $(k, t)$-VFTS having $O(\frac{k n}{(t-1)^{d-1}})$ edges in $O(\frac{1}{(t-1)^{d-1}}(k n \lg^{d} n + nk^{2} \lg k))$ time with total weight of edges upper bounded by $O(\frac{k^{2} \lg n}{(t-1)^{d}})$ multiplicative factor of the weight of a minimum spanning tree of the given set of points. 

\subsection*{Terminology}

Recall the Euclidean distance between two points $p$ and $q$ is denoted by $|pq|$ and the geodesic Euclidean distance between two points $p, q$ located in the free space of a polygonal domain is denoted by $d_\pi(p, q)$.
Here, $\pi$ denotes a shortest path between $p$ and $q$ located in the free space of the polygonal domain.
(Further, we would like to note the following terms are defined in previous subsection: polygonal domain, free space of a polygonal domain, and a geodesic shortest path between two points.)
We note that if the line segment joining $p$ and $q$ does not intersect any obstacle, $d_\pi(p, q)$ is equal to $|pq|$, otherwise $d_\pi(p, q)$ is the distance along a geodesic shortest path $\pi$ between $p$ and $q$.
The length of a shortest path between $p$ and $q$ in a graph $G$ is denoted by $d_G(p, q)$.
For a set $S'$ of vertices of $G$ and any two vertices $p, q$ of $G$ not belonging to $S'$, the distance along a shortest path between $p$ and $q$ in graph $G \setminus S'$ is denoted by $d_{G \backslash S'}(p, q)$.
As in \cite{journals/algorithmica/AbamBFGS11} and in \cite{conf/caldam/Inkulu19b}, the function $d_w$ is defined on a set $S$ of points as follows: for any $p, q \in S$, $d_w$ is equal to $w(p) + d_\pi(p, q) + w(q)$ if $p \ne q$; otherwise, $d_w(p, q)$ is equal to $0$. 

Recall that, for any set $S$ of points, any graph $G$ with vertex set $S$ and each of its edges associated with a non-negative weight is called a {\it $t$-spanner} of points in $S$ whenever $d_\pi(p, q) \le d_G(p, q) \le t \cdot d_\pi(p, q)$ for every two points $p, q \in S$ and a real number $t \ge 1$.
The smallest $t$ for which $G$ is a $t$-spanner of $S$ is called the {\it stretch factor} of $G$. 
The number of edges in $G$ is known as the size of $G$. 
A graph $G(S, E)$ is a {\it $k$-vertex fault-tolerant $t$-spanner}, denoted by $(k, t)$-VFTS, for a set $S$ of $n$ points whenever for any subset $S'$ of $S$ with size at most $k$, the graph $G \setminus S'$ is a $t$-spanner for the points in $S \setminus S'$.
For a real number $t > 1$ and a set $S$ of weighted points, a graph $G(S, E)$ is called a {\it $t$-spanner for weighted points in $S$} whenever $d_w(p, q) \le d_G(p, q) \le t \cdot d_w(p, q)$ for every two points $p$ and $q$ in $S$.
Given a set $S$ of points, a function $w$ to associate a non-negative weight to each point in $S$, an integer $k \geq 1$, and a real number $t > 0$, a geometric graph $G$ is called a {\it $(k, t, w)$-vertex fault-tolerant spanner for weighted points} in $S$, denoted by $(k, t, w)$-VFTSWP, whenever for any set $S' \subset S$ with cardinality at most $k$, the graph $G \setminus S'$ is a $t$-spanner for the weighted points in $S \setminus S'$.

\subsection*{Our Results}

The spanners computed in this paper are the first of their kind: we explore the fault-tolerance notion in the context of weighted points.
Given a set $S$ of $n$ points, a weight function $w$ to associate a non-negative weight to each point in $S$, a positive integer $k$, and a real number $\epsilon > 0$, we devise the following algorithms for computing vertex fault-tolerant spanners for the set $S$ of weighted points:

\begin{itemize}
\item[*]
When the points in $S$ are located in $\mathbb{R}^d$, we compute a $(k, 4+\epsilon)$-VFTSWP with $O(k n)$ edges.
We incorporate fault-tolerance to the recent results of \cite{conf/soda/AbamBS17}.
In specific, the stretch factor of the spanner computed here is same as the spanner computed in \cite{conf/soda/AbamBS17}, while the number of edges in our spanner is $O(k)$ times the number of edges of the spanner computed in \cite{conf/soda/AbamBS17}.

\item[]

\item[*]
When the points in $S$ are located in the relative interior of a polygonal domain $\cal P$, we present an algorithm to compute a $(k, 4+\epsilon)$-VFTSWP with $O(\frac{kn\sqrt{h+1}}{\epsilon^{2}}\lg{n})$ edges.
Here, $h$ is the number of holes in $\cal P$.
We extend the algorithms in \cite{conf/compgeom/AbamAHA15} to achieve the vertex fault-tolerance.

\item[]

\item[*]
When the points in $S$ are located on a polyhedral terrain, an algorithm is presented to compute a $(k, 4+\epsilon)$-VFTSWP with $O(\frac{kn}{\epsilon^{2}}\lg{n})$ edges.
By extending the result in \cite{conf/soda/AbamBS17}, we achieve the vertex fault-tolerance for the spanner among weighted points. 
\end{itemize}

Section~\ref{sect:rd} details the algorithm and its analysis to compute a $(k, 4+\epsilon)$-VFTSWP when the weighted points are in $\mathbb{R}^d$.
For the case of weighted points located in the free space of any polygonal domain, Section~\ref{sect:polydom} details an algorithm to compute a $(k, 4+\epsilon)$-VFTSWP.
Section~\ref{sect:terrains} presents an algorithm to compute a $(k, 4+\epsilon)$-VFTSWP when the input points are located on a polyhedral terrain.
Conclusions are in Section~\ref{sect:conclu}.

\vspace{-0.12in}
\section{Vertex fault-tolerant spanner for weighted points in $\mathbb{R}^d$}
\label{sect:rd}

Given a set $S$ of $n$ points located in $\mathbb{R}^d$, a weight function $w$ to associate a non-negative weight to each point in $S$, a positive integer $k$, and a real number $\epsilon > 0$, we devise an algorithm to compute a $(k, 4+\epsilon, w)$-vertex fault-tolerant spanner for the set $S$ of weighted points.
For any two distinct points $p, q \in S$, each associated with a non-negative weight, the weighted distance between $p$ and $q$, denoted by $d_{w}(p, q)$, is defined as the $w(p) + |pq| + w(q)$.
Following the algorithm in \cite{conf/soda/AbamBS17}, we partition the set $S$ of points into clusters.
As part of this clustering, the points in set $S$ are sorted in non-decreasing order of their weights.
The first $k+1$ points in this sorted list are chosen as the centers of $k+1$ distinct clusters. 
As the algorithm progress, more points may be added to these clusters as well as more clusters (with respective cluster centers) may also be initiated.

Our algorithm considers points in $S$ in non-decreasing order with respect to their weights.
For each point $p$ considered, among the current set of cluster centers, we determine a cluster center $c_j$ that is nearest to $p$.
Let $C_j$ be the cluster to which $c_j$ is the center.
It adds $p$ to the cluster $C_j$ if $|p c_j| \leq \epsilon \cdot w(p)$; otherwise, a new cluster with $p$ as its center is initiated. 
We note that choosing cluster corresponding to the nearest cluster center to each point helps in upper bounding the stretch factor.
Let $C = \{ c_1, \ldots, c_z \}$ be the final set of cluster centers obtained in this algorithm.
For every $i \in [1, z]$, we denote the cluster to which $c_i$ is the center with $C_i$.
Using the algorithm from \cite{conf/stoc/Solomon14}, we compute a $(k, 2+\epsilon)$-VFTS $\mathcal{B}$ for the set $C$ of cluster centers.
We note that the degree of each vertex of $\mathcal{B}$ is $O(k)$.
We denote the stretch of $\cal{B}$ by $t_{\cal{B}}$.

The spanner graph $G$ being constructed is initialized so that its vertex set comprises of points in $S$ and its edge set comprises of all the edges in $\cal B$.
Our algorithm to compute a $(k, 4+\epsilon)$-VFTSWP differs from \cite{conf/soda/AbamBS17} with respect to both the algorithm used in computing $\mathcal{B}$ and the edges included in $G$. 
The latter part is described now.
For every $i \in [1, z]$, let $C_i'$ be the set comprising of any $\min\{k+1, |C_i|\}$ least weighted points of cluster $C_i$. 
Also, for every $i \in [1, z]$, let $B_i$ be the set comprising of all the adjacent vertices of the node corresponding to center ($c_i$) of cluster $C_i$ in $\cal B$.
In specific, let $B_l$ be the set comprising of all the adjacent vertices of the node corresponding to center ($c_l$) of cluster $C_l$ in $\cal B$.
For each point $p \in S \setminus C$, if $p$ belongs to cluster $C_l$, then for each $v \in B_l \cup C_l'$, our algorithm introduces an edge between $p$ and $v$ with weight $|pv|$ to $G$. 
These additional edges (one not part of $\cal B$) of $G$ help in achieving the $k$-fault tolerance.
Refer to Algorithm~\ref{alg:addfts}.

\begin{algorithm}[ht]
    \caption{VFTSWPRd($S, k, \epsilon$)}
    \label{alg:addfts}
     
    \SetAlgoLined
    \SetKwInOut{Input}{Input}
    \SetKwInOut{Output}{Output}
    
    \Input{A set $S$ of $n$ points located in $\mathbb{R}^d$, a weight function $w$ that associates a non-negative weight to each point in $S$, an integer $k \geq 1$, and a real number $\epsilon > 0$.}
    \Output{A $(k,(4+\epsilon))$-VFTSWP $G(S, E)$ for the points in $S$.}

    Sort points in $S$ in non-decreasing order of their weight, and save them in that order in an array $L$. \\

    $C \leftarrow \phi$ \hspace{0.1in} \scriptsize{} (The set $C$ stores the centers of clusters.)\normalsize{} \\

    \ForEach{$i$ from $1$ to $k+1$} {
        Initialize a cluster $C_i$ with its center $c_i = L[i]$.  \\
        $C \leftarrow C \cup \{c_i\}$.
    }

    $z \leftarrow k+1$ \\
    
    \ForEach{$i$ from $k+2$ to $n$}{
        $p \leftarrow L[i]$. \\
    
        Among all the cluster centers in $C$, find the center $c_j$ that is at a minimum distance from $p$. \\

        \uIf{$|p c_j| \leq \epsilon. w(p)$}{
            $C_j \leftarrow C_j \cup \{ p \}$.
            
        } \uElse{
            $z \leftarrow z+1$. \\

            Initialize a new cluster $C_{z}$ with its center $c_z = p$. \\
            
            $C \leftarrow C \cup \{c_z\}$.
        }
    }
        
    Using the algorithm in \cite{conf/stoc/Solomon14}, construct a $(k,(1+\epsilon))$-VFTS $\mathcal{B}$ for the cluster centers in $C$.
    
    \ForEach{$p \in S$ and $p \notin C$}{
        Find the cluster $C_i$ to which $p$ belongs; let $c_i$ be the center of $C_i$. \\

        Let $C_i'$ be the set of $min \{k+1, |C_i|\}$ least weighted points of cluster $C_i$.  For every $p' \in C_i'$, introduce an edge in $G$ between $p$ and $p'$ with weight $|pp'|$. \\
        
        Let $B_i$ be the $k$ nearest neighbors of $c_i$ in $\cal B$.  For every $v \in B_i$, introduce an edge in $G$ between $p$ and $v$ with weight $|pv|$.
    }
\end{algorithm}

With the help of the following facts, we prove the graph $G$ has $4+\epsilon$ stretch for weighted points in $S$:
(i) the network $\cal B$ is a $(k, t_{\cal B})$-VFTS for the cluster centers in $C$,
(ii) the weight associated with every point in $S$ is non-negative,
(iii) points in $S$ are sorted in non-decreasing order of their weights as part of initializing the clusters, 
(iv) the first point added to any cluster is the center of that cluster,
(v) any point $x$ is included into a cluster $C_l$ only if $|xc_l| \le \epsilon \cdot w(x)$, and
(v) the triangle inequality in $\mathbb{R}^d$.
In the following theorem, by choosing $t_{\cal B} = 2+\epsilon$, we prove that the graph $G$ is indeed a $(k, 4+\epsilon)$-VFTSWP and it has $O(kn)$ edges.

\begin{theorem}
\label{thm:rd}
Given a set $S$ of $n$ points in $\mathbb{R}^d$, a weight function $w$ to associate a non-negative weight to each point in $S$, a positive integer $k$, and a real number $\epsilon > 0$, Algorithm~\ref{alg:addfts} computes a $(k, 4+\epsilon)$-vertex fault tolerant spanner with $O(kn)$ edges for the weighted points in $S$.
\end{theorem}

\begin{proof}
From \cite{conf/stoc/Solomon14}, the number of edges in $\mathcal{B}$ is $O(k \hspace{0.02in} |C|)$, which is $O(kn)$.
Further, the degree of each node in $\mathcal{B}$ is $O(k)$.
For each point in $p \in S \setminus C$, we are introducing $O(k)$ edges that are incident to $p$. 
Hence, the number of edges in $G$ is $O(k n)$.

In proving $G$ is a $(k, 4+\epsilon)$-VFTSWP for the metric space $(S, d_w)$, we show that for any set $S' \subset S$ with $|S'| \leq k$ and for any two points $p, q \in S \setminus S'$, there exists a path between $p$ and $q$ in $G \setminus S'$ such that the distance along that path is at most $(4+\epsilon) \cdot d_w(p, q)$.
Based on the construction, it is immediate to note that $d_{G \setminus S'}(p, q)$ is at least $d_w(p, q)$.
The following cases are exhaustive, and these cases are based on the relative location of points $p$ and $q$ with respect to clusters. \\

\noindent
{\it Case 1}: Both $p$ and $q$ are cluster centers of two different clusters, i.e., $p,q \in C$.

\noindent
Since $\mathcal{B}$ is a $(k, 2+\epsilon)$-VFTS for the cluster centers, and since edges in $\cal B$ are included in $G$, 
$d_{G \setminus S'}(p,q) \le d_{\mathcal{B} \setminus S'}(p,q) \leq t_{\cal B} \cdot d_w(p,q)$.

\hfil\break

\noindent
{\it Case 2}: Both $p$ and $q$ are in the same cluster $C_i$, and one of them, without loss of generality, say $p$, is the center of $C_i$.

\noindent
Since $p$ is the least weighted point in $C_i$, there exists an edge between $p$ and $q$ in $G$.
Hence, $d_{G \setminus S'}(p,q) = d_w(p,q)$.

\hfil\break

\noindent
{\it Case 3}: Both $p$ and $q$ are in the same cluster, say $C_i$, neither $p$ nor $q$ is the center of $C_i$, and, $c_i \notin S'$.

\noindent
Then, 
{\setlength{\abovedisplayskip}{0pt}
\begin{flalign}
d_{G \setminus S'}(p,q) &= d_w(p,c_i) + d_w(c_i, q)&&\nonumber\\
           &= w(p) + |p c_i| + w(c_i) + w(c_i) + |c_i q| + w(q)&&\nonumber\\
            &\leq w(p) + \epsilon \cdot w(p) + w(c_i) + w(c_i) + \epsilon \cdot w(q) + w(q)&&\nonumber\\
            &[\text{since a point} \ x \ \text{is added to cluster} \ C_l \ \text{only if} \ |x c_l| \leq \epsilon \cdot w(x)]&&\nonumber\\
            &\leq w(p) + \epsilon \cdot w(p) + w(p) + w(q) + \epsilon \cdot w(q) + w(q)&&\nonumber\\
            &[\text{since the points are sorted in the non-decreasing order of their weights, and since the first}&&\nonumber\\ 
            &\text{point added to any cluster is the center of that cluster}]&&\nonumber
\end{flalign}}

{\setlength{\abovedisplayskip}{0pt}
\begin{flalign}
\hspace{14mm}
	    &= (2 + \epsilon) \cdot [w(p) + w(q)]&&\nonumber\\
            &< (2 + \epsilon) \cdot [w(p) + |pq| + w(q)]&&\nonumber\\
	    &= (2 + \epsilon) \cdot d_w(p,q).&\nonumber
\end{flalign}}

\noindent
{\it Case 4}: Both $p$ and $q$ are in the same cluster, say $C_i$; $p \ne c_i$, $q \ne c_i$; and, $c_i \in S'$.

\noindent
In the case of $|C_i| \leq k$, there exists an edge between $p$ and $q$ in $G$.
Hence, suppose that $|C_i| > k$.
Let $C_i'$ be the set consisting of at most $k+1$ least weighted points from $C_i$.
If $p,q \in C_i'$, there exists an edge between $p$ and $q$ in $G$. 
Also, if $p \in C_i'$ (resp. $q \in C_i'$) and $q \notin C_i'$ (resp. $p \notin C_i'$), there exists an edge between $p$ and $q$ in $G$.
Now consider the case in which both $p, q \notin C_i'$.
Since there is an edge between every point in $C_i'$ and every point in set $\{p, q\}$, there exists an $r \in C_i'$ such that $r \notin S'$ and the edges $(p, r)$ and $(r, q)$ belong to $G \setminus S'$.
Therefore,
{\setlength{\abovedisplayskip}{0pt}
\begin{flalign}
d_{G \setminus S'}(p,q) &= d_w(p,r) + d_w(r, q)&&\nonumber\\
            &= w(p) + |p r| + w(r) + w(r) + |r q| +w(q)&&\nonumber\\
            &\leq w(p) + |p c_i| + |c_i r| + w(r) + w(r) + |r c_i| + |c_i q| + w(q)&&\nonumber\\
            &\text{[by triangle inequality]}&&\nonumber\\
	&\leq w(p) + \epsilon \cdot w(p) + \epsilon \cdot w(r) + w(r) + w(r) + \epsilon \cdot w(r) + \epsilon \cdot w(q) + w(q)&&\nonumber\\
             &[\text{since a point} \ x \ \text{is added to cluster} \ C_l \ \text{only if} \ |x c_l| \leq \epsilon \cdot w(x)]&&\nonumber\\
            &\leq w(p) + \epsilon \cdot w(p) + \epsilon \cdot w(p) + w(p) + w(q) + \epsilon \cdot w(q) + \epsilon \cdot w(q) + w(q)&&\nonumber\\
	    &\text{[since $r \in C_i'$ and $p, q \notin C_i'$]}&&\nonumber\\
             &= (2 + 2\epsilon) \cdot [w(p) + w(q)]&&\nonumber\\
	     &< (2 + 2\epsilon) \cdot [w(p) + |pq| + w(q)]&&\nonumber\\
             &= (2+2\epsilon) \cdot d_w(p,q).&\nonumber
\end{flalign}}

\noindent
{\it Case 5}: The points $p$ and $q$ belong to two distinct clusters, say $p \in C_i$ and $q \in C_j$. 
In addition, $p \neq c_i$ and $q \ne c_j$, and neither $c_i$ nor $c_j$ belongs to $S'$.

\noindent
Then,
{\setlength{\abovedisplayskip}{0pt}
\begin{flalign}
	d_{G \setminus S'}(p,q) &= d_w(p,c_i) + d_{\mathcal{B}}(c_i, c_j) + d_w(c_j, q)&&\nonumber\\
            &= w(p) + |p c_i| + w(c_i) + d_{\mathcal{B}}(c_i, c_j) + w(c_j) + |c_j q| +w(q)&&\nonumber\\
            &\leq w(p) + \epsilon \cdot w(p) + w(c_i) + d_{\mathcal{B}}(c_i, c_j) + w(c_j) + \epsilon \cdot w(q) + w(q)&&\nonumber\\
            &[\text{since a point} \ x \ \text{is added to cluster} \ C_l \ \text{only if} \ |x c_l| \leq \epsilon \cdot w(x)]&&\nonumber\\
            &\leq (1 + \epsilon) \cdot [w(p) + w(q)] + w(c_i) + w(c_j) + t_{\mathcal{B}} \cdot d_w(c_i,c_j)&&\nonumber\\
            &[\text{since} \ \mathcal{B} \ \text{is a} \ (k, t_{\mathcal{B}})\text{-VFTS for the cluster centers in $C$}]&&\nonumber\\
            &\leq (1 + \epsilon) \cdot [w(p) + w(q)] + w(p) + w(q) + t_{\mathcal{B}} \cdot d_w(c_i,c_j)&&\nonumber\\
            &\text{[since the points are sorted in the non-decreasing order of their weights, and since the first}&&\nonumber\\ 
            &\text{point added to any cluster is the center of that cluster]}&&\nonumber\\
            &= (2 + \epsilon) \cdot [w(p) + w(q)] + t_{\mathcal{B}} \cdot [w(c_i) + |c_i c_j| + w(c_j)]&&\nonumber\\
            &\leq (2 + \epsilon) \cdot [w(p) + w(q)] + t_{\mathcal{B}} \cdot [w(p) + |c_i c_j| + w(q)]&&\nonumber\\
            &\text{[since the points are sorted in the non-decreasing order of their weights, and since the first}&&\nonumber\\ 
            &\text{point added to any cluster is the center of that cluster]}&&\nonumber\\
            &\leq (2 + \epsilon) \cdot [w(p) + w(q)] + t_{\mathcal{B}} \cdot [w(p) + w(q) + |c_i p| + |pq| + |q c_j|]&&\nonumber\\
            &\text{[by triangle inequality]}&&\nonumber
\end{flalign}}

{\setlength{\abovedisplayskip}{0pt}
\begin{flalign}
\hspace{12mm}
            &\leq (2 + \epsilon) \cdot [w(p) + w(q)] + t_{\mathcal{B}} \cdot [w(p) + w(q) + \epsilon \cdot w(p) + |pq| + \epsilon \cdot w(q)]&&\nonumber\\ &[\text{since a point} \ x \ \text{is added to cluster} \ C_l \ \text{only if} \ |x c_l| \leq \epsilon \cdot w(x)]&&\nonumber\\
            &= (2 + \epsilon) \cdot [w(p) + w(q)] + t_{\mathcal{B}} \cdot [(1 + \epsilon) \cdot [w(p) + w(q)] + |pq|]&\nonumber\\
           &< (2 + \epsilon) \cdot [w(p) + w(q) + |pq|] + t_{\mathcal{B}} \cdot (1 + \epsilon) \cdot [w(p) + w(q) + |pq|]&&\nonumber\\
	   &< t_{\mathcal{B}}(2 + \epsilon) \cdot [w(p) + w(q) + |pq|] \hspace{0.1in} \text{ when } t_B \ge (2+\epsilon) &&\nonumber\\
	   &\text{[since the weight associated with every point is non-negative]}&&\nonumber\\
	     &= t_{\mathcal{B}} \cdot (2 + \epsilon) \cdot d_w(p,q).&\nonumber
\end{flalign}}

\noindent
{\it Case 6}: Both the points $p$ and $q$ are in two distinct clusters, without loss of generality, say $p \in C_i$ and $q \in C_j$, and without loss of generality, $p$ is the center of $C_i$ (i.e., $p = c_i$) and $c_j \notin S'$.  Then,
{\setlength{\abovedisplayskip}{0pt}
\begin{flalign}
	d_{G \setminus S'}(p,q) &\le d_{\mathcal{B}}(c_i,c_j) + d_w(c_j, q)&&\nonumber\\
            &\leq t_{\mathcal{B}} \cdot d_w(c_i, c_j) + d_w(c_j, q)&&\nonumber\\
        &[\text{since} \ \mathcal{B} \ \text{is a} \ (k, t_{\mathcal{B}})\text{-VFTS for the set of cluster centers in $C$}]&&\nonumber\\
            &= t_{\mathcal{B}} \cdot d_w(c_i, c_j) + w(c_j) + |c_j q| + w(q)&&\nonumber\\
            &= t_{\mathcal{B}} \cdot [w(c_i) + |c_i c_j| + w(c_j)] + w(c_j) + |c_j q| + w(q)&&\nonumber\\
            &\leq t_{\mathcal{B}} \cdot [w(p) + |c_i c_j| + w(q)] + w(q) + |c_j q| + w(q)&&\nonumber\\
            &\text{[since the points are sorted in the non-decreasing order of their weights, and since the first}&&\nonumber\\ 
            &\text{point added to any cluster is the center of that cluster]}&&\nonumber\\
        &\leq t_{\mathcal{B}} \cdot [w(p) + |c_i c_j| + w(q)] + w(q) + \epsilon \cdot w(q) + w(q)&&\nonumber\\
	&[\text{since any point} \ x \ \text{is added to cluster} \ C_l \ \text{only if} \ |x c_l| \leq \epsilon \cdot w(x)]&&\nonumber\\
            &\leq t_{\mathcal{B}} \cdot [w(p) + |c_i p| + |pq| + |q c_j| + w(q)] + w(q) + \epsilon \cdot w(q) + w(q)&&\nonumber\\
            &\text{[by triangle inequality]}&&\nonumber\\ 
            &\leq t_{\mathcal{B}} \cdot [w(p) + \epsilon \cdot w(p) + |pq| + \epsilon \cdot w(q) + w(q)] + w(q) + \epsilon \cdot w(q) + w(q)&&\nonumber\\
            &[\text{since any point} \ x \ \text{is added to cluster} \ C_l \ \text{only if} \ |x c_l| \leq \epsilon.w(x)]&&\nonumber\\
            &= t_{\mathcal{B}} \cdot [(1 + \epsilon) \cdot [w(p) + w(q)] + |pq|] + (2 + \epsilon) \cdot w(q)&&\nonumber\\
            &\leq t_{\mathcal{B}} \cdot [(1 + \epsilon) \cdot [w(p) + w(q)] + |pq|] + (2 + \epsilon) \cdot [w(p) + w(q) + |pq|]&&\nonumber\\
	&\text{[since the weight associated with each point is non-negative]}&&\nonumber\\
        &\leq t_{\mathcal{B}} \cdot (2 + \epsilon) \cdot [w(p) + w(q) + |pq|] \hspace{0.1in} \text{ when } t_{\mathcal{B}} \ge (2+\epsilon) &&\nonumber\\
            &\leq t_{\mathcal{B}} \cdot (2 + \epsilon) \cdot d_w(p,q).&\nonumber
\end{flalign}}

\noindent
{\it Case 7}: Both the points $p$ and $q$ are in two distinct clusters, say $p \in C_i$ and $q \in C_j$; $p \neq c_i$, $q \neq c_j$; and, one of these centers, say $c_j$, belongs to $S'$ and the other center $c_i \notin S'$.

\noindent
Since there is an edge between $q$, and since each of the neighbors of $c_j$ is in $\mathcal{B}$, for any neighbor $c_r$ of $c_j$ with $c_r \in C$ and $c_r \notin S'$, edge $(q, c_r) $ belongs to $G \setminus S'$.  
Therefore, 
{\setlength{\abovedisplayskip}{0pt}
\begin{flalign}
	d_{G \setminus S'}(p,q) &= d_w(p,c_i) + d_{\mathcal{B}}(c_i, c_r) + d_w(c_r, q)&&\nonumber\\
        &= w(p) + |p c_i| + w(c_i) + d_{\mathcal{B}}(c_i, c_r) + w(c_r) + |c_r q| + w(q)&\nonumber\\
        &\leq w(p) + \epsilon \cdot w(p) + w(c_i) + d_{\mathcal{B}}(c_i, c_r) + w(c_r) + |c_r q| + w(q)&&\nonumber\\
        &[\text{since a point} \ x \ \text{is added to cluster} \ C_l \ \text{only if} \ |x c_l| \leq \epsilon \cdot w(x)]&&\nonumber\\
        &\leq w(p) + \epsilon \cdot w(p) + w(c_i) + d_{\mathcal{B}}(c_i, c_r) + w(c_r) + |c_r c_j| + |c_j q| + w(q)&&\nonumber
\end{flalign}}

{\setlength{\abovedisplayskip}{0pt}
\begin{flalign}
\hspace{14mm}
        &\text{[by triangle inequality]}&\nonumber\\
	&\leq w(p) + \epsilon \cdot w(p) + w(c_i) + d_{\mathcal{B}}(c_i, c_r) + w(c_r) + |c_r c_j| + \epsilon \cdot w(q) + w(q)&&\nonumber\\
        &[\text{since a point} \ x \ \text{is added to cluster} \ C_l \ \text{only if} \ |x c_l| \leq \epsilon \cdot w(x)]&&\nonumber\\
        &\leq (1 + \epsilon) \cdot [w(p) + w(q)] + w(c_i) + d_{\mathcal{B}}(c_i, c_r) + w(c_r) + |c_j c_r|\label{eq1}.&
\end{flalign}}

\noindent
Since $\mathcal{B}$ is a $(k, t_{\mathcal{B}})$-VFTS, there are at least $k+1$ vertex disjoint paths between $c_i$ and $c_j$ in $\cal B$ such that the distance along each of these paths is at most $t_{\mathcal{B}} \cdot d_w(c_i, c_j)$.
Suppose $c_r$ is the neighbor of $c_j$ in $\mathcal{B}$ so that one of these $k+1$ paths between $c_i$ to $c_j$ pass through $c_r$.
Then, we have the following:
{\setlength{\abovedisplayskip}{0pt}
\begin{flalign}
	&d_{\mathcal{B}}(c_i , c_r) + d_w(c_r , c_j) < t_{\mathcal{B}} \cdot d_w(c_i , c_j)&&\nonumber\\
	     &\Rightarrow d_{\mathcal{B}}(c_i , c_r) + w(c_r) + |c_r  c_j| + w(c_j) < t_{\mathcal{B}} \cdot d_w(c_i , c_j)&&\nonumber\\
	     &\Rightarrow d_{\mathcal{B}}(c_i , c_r) + |c_r c_j| + w(c_r) < t_{\mathcal{B}} \cdot d_w(c_i , c_j) - w(c_j).\label{eq2}&
\end{flalign}}

\noindent
Substituting (\ref{eq2}) in (\ref{eq1}),
{\setlength{\abovedisplayskip}{0pt}
\begin{flalign}
d_{G \setminus S'}(p,q) &< (1 + \epsilon) \cdot [w(p) + w(q)] + w(c_i) + t_{\mathcal{B}} \cdot d_w(c_i , c_j) - w(c_j)&&\nonumber\\
        &\leq (1 + \epsilon) \cdot [w(p) + w(q)] + w(c_i) + t_{\mathcal{B}} \cdot d_w(c_i , c_j)&&\nonumber\\
        &\text{[since the weight associated with every point in $S$ is non-negative]}&&\nonumber\\
        &=(1 + \epsilon) \cdot [w(p) + w(q)] + w(c_i) + t_{\mathcal{B}} \cdot [w(c_i) + |c_i c_j| + w(c_j)]&&\nonumber\\
        &\leq (1 + \epsilon) \cdot [w(p) + w(q)] + w(p) + t_{\mathcal{B}} \cdot [w(p) + |c_i c_j| + w(q)]&&\nonumber\\
        &\text{[since the points are sorted in the non-decreasing order of their weights, and since the first}&&\nonumber\\
        &\text{point added to any cluster is taken as its center]}&&\nonumber\\
	&\leq (1 + \epsilon) \cdot [w(p) + w(q)] + w(p) + t_{\mathcal{B}} \cdot [w(p) + |c_i p| + |pq| + |q c_j| + w(q)]&&\nonumber\\
        &\text{[by triangle inequality]}&&\nonumber\\
        &\leq (1 + \epsilon) \cdot [w(p) + w(q)] + w(p) + t_{\mathcal{B}} \cdot [w(p) + \epsilon \cdot w(p) + |pq| + \epsilon \cdot w(q) + w(q)]&&\nonumber\\
        &[\text{since a point} \ x \ \text{is added to cluster} \ C_l \ \text{only if} \ |x c_l| \leq \epsilon \cdot w(x)]&&\nonumber\\
        &\leq (2 + \epsilon) \cdot [w(p) + w(q)] + t_{\mathcal{B}} \cdot [(1 + \epsilon) \cdot [w(p) + w(q)] + |pq|]&&\nonumber\\
    &\leq t_{\mathcal{B}} \cdot [(2 + \epsilon) \cdot [w(p) + w(q)] + |pq|] \hspace{0.1in} \text{ when } t_{\mathcal{B}} \ge (2+\epsilon) &&\nonumber\\
        &\leq t_{\mathcal{B}} \cdot (2 + \epsilon) \cdot [w(p) + w(q)] + t_{\mathcal{B}}[w(p) + |pq| + w(q)]&&\nonumber\\
        &\text{[since the weight associated with every point in $S$ is non-negative]}&&\nonumber\\
        &\leq t_{\mathcal{B}} \cdot (2 + \epsilon) \cdot d_w(p,q).&\nonumber
\end{flalign}}

\noindent
{\it Case 8}: Both $p$ and $q$ are in two distinct clusters, say $p \in C_i$ and $q \in C_j$; $p \neq c_i$, $q \neq c_j$; and, both $c_i, c_j \in S'$.

\noindent
Since there is an edge between $p$ (resp. $q$) and each of the $k$ nearest neighbors of both $c_i$ and $c_j$, there exists two points $c_r, c_l \in C$ such that $c_r, c_l \notin S'$ and the edges $(p, c_r)$ and $(c_l, q)$ belong to $G \setminus S'$.  
Therefore, 
{\setlength{\abovedisplayskip}{0pt}
\begin{flalign}
	d_{G \setminus S'}(p,q) &= d_w(p,c_r) + d_{\mathcal{B}}(c_r, c_l) + d_w(c_l, q)&&\nonumber\\
        &= w(p) + |p c_r| + w(c_r) + d_{\mathcal{B}}(c_r, c_l) + w(c_l) + |c_l q| + w(q)&\nonumber\\
	&\leq w(p) + |p c_i| + |c_i c_r| + w(c_r) + d_{\mathcal{B}}(c_r, c_l) + w(c_l) + |c_j c_l| + |c_j q| + w(q)\label{eq3}&&\\
        &\text{[by triangle inequality]}.&\nonumber
\end{flalign}}

\noindent
Since $\mathcal{B}$ is a $(k, t_{\mathcal{B}})$-VFTS, there are at least $k+1$ vertex disjoint paths between $c_i$ and $c_j$ in $\cal B$ such that the distance along each of these paths is at most $t_{\cal B} \cdot d_w(c_i, c_j)$. 
Suppose $c_r$ (resp. $c_l$) is the neighbor of $c_i$ (resp. $c_j$) in $\mathcal{B}$ so that one of these $k+1$ paths between $c_i$ and $c_j$ passes through $c_r$ (resp. $c_l$).
Then, we have the following:
{\setlength{\abovedisplayskip}{0pt}
\begin{flalign}
	&d_w(c_i,c_r) + d_{\mathcal{B}}(c_r,c_l) + d_w(c_l,c_j) < t_{\mathcal{B}} \cdot d_w(c_i,c_j)&&\nonumber\\
	&\Rightarrow w(c_r) + |c_i c_r| + w(c_i) + d_{\mathcal{B}}(c_r , c_l) + w(c_j) + |c_j c_l| + w(c_l) < t_{\mathcal{B}} \cdot d_w(c_i , c_j)&&\nonumber\\
	&\Rightarrow w(c_r) + |c_i c_r| + d_{\mathcal{B}}(c_r , c_l) + w(c_l) + |c_j c_l|  < t_{\mathcal{B}} \cdot d_w(c_i , c_j) - w(c_i) - w(c_j).\label{eq4}&
\end{flalign}}

\noindent
Substituting (\ref{eq4}) in (\ref{eq3}),
{\setlength{\abovedisplayskip}{0pt}
\begin{flalign}
d_{G \setminus S'}(p,q) &< w(p) + |p c_i| + t_{\mathcal{B}} \cdot d_w(c_i , c_j) - w(c_i) - w(c_j) + |c_j q| + w(q)&&\nonumber\\
        &\leq w(p) + |p c_i| + t_{\mathcal{B}} \cdot d_w(c_i , c_j) + |c_j q| + w(q)&&\nonumber\\
        &\text{[since the weight associated with every point in $S$ is non-negative]}&&\nonumber\\
        &= w(p) + |p c_i| + t_{\mathcal{B}} \cdot [w(c_i) + |c_i c_j| + w(c_j)] + |c_j q| + w(q)&&\nonumber\\
        &\leq w(p) + |p c_i| + t_{\mathcal{B}} \cdot [w(c_i) + |c_i p| + |pq| + |q c_j| + w(c_j)] + |c_j q| + w(q)&&\nonumber\\
	&\text{[by triangle inequality]}&&\nonumber\\
	&\leq w(p) + \epsilon \cdot w(p) + t_{\mathcal{B}} \cdot [w(c_i) + \epsilon \cdot w(p) + |pq| + \epsilon \cdot w(q) + w(c_j)] + \epsilon \cdot w(q) + w(q)&&\nonumber\\
        &[\text{since a point} \ x \ \text{is added to cluster} \ C_l \ \text{only if} \ |x c_l| \leq \epsilon \cdot w(x)]&&\nonumber\\
	&\leq w(p) + \epsilon \cdot w(p) + t_{\mathcal{B}} \cdot [w(p) + \epsilon \cdot w(p) + |pq| + \epsilon \cdot w(q) + w(q)] + \epsilon \cdot w(q) + w(q)&&\nonumber\\
        &[\text{since the points are sorted in the non-decreasing order of their weights, and since the first}&&\nonumber\\
        &\text{point added to any cluster is the center of that cluster}]&&\nonumber\\
        &\leq (1 + \epsilon) \cdot [w(p) + w(q)] + t_{\mathcal{B}} \cdot [(1 + \epsilon) \cdot [w(p) + w(q)] + |pq|]&&\nonumber\\
	&\leq t_{\mathcal{B}} \cdot (2 + \epsilon) \cdot [w(p) + w(q) + |pq|] \hspace{0.1in} \text{ when } t_{\mathcal{B}} \ge (1+\epsilon) &&\nonumber\\
        &= t_{\mathcal{B}} \cdot (2 + \epsilon) \cdot d_w(p,q).&\nonumber
\end{flalign}}

\noindent
The analysis of these cases shows that $G$ is a $k$-VFTSWP with its stretch $t$ upper bounded by $t_{\mathcal{B}} \cdot (2 + \epsilon)$.
We choose $t_{\mathcal{B}}$ to be equal to $(2 + \epsilon)$, so that it satisfies all the above cases.
Since $t_{\mathcal{B}}$ is $(2 + \epsilon)$, $t = (2 + \epsilon)^2 \le (4 + 5 \epsilon)$.
Hence, $G$ is a $(k, 4 + 5\epsilon)$-VFTSWP for the metric space $(S, d_w)$.
\end{proof}

\vspace{-0.2in}
\section{Vertex fault-tolerant spanner for weighted points located in a polygonal domain}
\label{sect:polydom}

Given a polygonal domain $\cal P$, a set $S$ of $n$ points located in the free space $\cal D$ of $\cal P$, a weight function $w$ to associate a non-negative weight to each point in $S$, a positive integer $k$, and a real number $\epsilon > 0$, in this section we describe an algorithm to compute a geodesic $(k, 4+\epsilon)$-vertex fault-tolerant spanner for the set $S$ of weighted points. 
Recall that the polygonal domain $\cal{P}$ consists of a simple polygon $P$ and $h$ simple polygonal holes located interior to $P$.
Our algorithm extends the algorithm given in \cite{conf/compgeom/AbamAHA15}. 
(The result in \cite{conf/compgeom/AbamAHA15} computes a $(5+\epsilon)$-spanner for the given set of unweighted points lying in the free space of a polygonal domain.)

Our algorithm decomposes the free space $\mathcal{D}$ into $O(h)$ simple polygons.
This is accomplished by introducing a line segment between the leftmost (resp. rightmost) vertex (along the $x$-axis) of every polygonal hole in $\cal P$ and its point of projection along the positive $y$-axis.
Analogously, we introduce a line segment between the leftmost (resp. rightmost) vertex (along the $x$-axis) of every polygonal hole in $\cal P$ and its point of projection along the negative $y$-axis.
Each of these line segments is called a {\it splitting segment}.
Note that the number of splitting segments is $O(h)$.
If any of the resulting simple polygons has more than three splitting segments on its boundary, then that simple polygon is further decomposed.
This algorithm results in partitioning $\mathcal{D}$ into $O(h)$ simple polygons.
Then the dual graph $G_d$ corresponding to this simple polygonal subdivision is computed:
each vertex of $G_d$ corresponds to a unique simple polygon, and each edge in $G_d$ corresponds to two simple polygons sharing a splitting segment. 
Each vertex $v$ of $G_d$ is associated with a weight equal to the number of points that lie inside the simple polygon corresponding to $v$.
We note that $G_d$ is a planar graph.
We use the following theorem from \cite{journals/siamdm/AlonST94} to compute a $O(\sqrt{h})$-separator for the planar graph $G_d$.

\begin{theorem}
\label{thm4}
Suppose $G=(V,E)$ is a planar vertex-weighted graph with $|V| = m$.
Then, an $O(\sqrt{m})$-separator for $G$ can be computed in $O(m)$ time. 
That is, $V$ can be partitioned into sets $R', R''$ and $R$ such that $|R| = O(\sqrt{m})$, there is no edge between $R'$ and $R''$, and $w(R'),w(R'') \leq \frac{2}{3}w(V)$.
Here, $w(X)$ is the sum of weights of all vertices in $X$.
\end{theorem}

For any point $p$ located in the free space $\cal D$ of a polygonal domain, and for any line segment $l$ located in $\cal D$, any point $p_l \in l$ that is at the minimum geodesic distance from $p$ among all the points of $l$ is called a {\it geodesic projection of $p$ on $l$}.
The geodesic distance between $p$ and $p_l$ is called the {\it geodesic distance between $p$ and $l$}.
We denote a geodesic projection of a point $p$ on a line segment $l$ with $p_l$.
Note that geodesic projection of a point on a line segment need not be unique.
However, if the polygonal domain has no holes, then it is unique.

Let $R', R'',$ and $R$ be the sets into which the vertices of $G_d$ is partitioned due to $O(\sqrt{h})$ sized vertex separator $R$ computed using Theorem~\ref{thm4}.
For each vertex $r \in R$, we collect the splitting segments of the simple polygon corresponding to $r$ into a set $H$, i.e., $O(\sqrt{h})$ splitting segments are collected into a set $H$. 
For each splitting segment $l$ in $H$, we proceed as follows.
For each point $p \in S$, we find a geodesic projection $p_l$ of $p$ on $l$; we assign the weight $w(p) + d_\pi(p, p_l)$ to point $p_l$ and include $p_l$ into the set $S_l$ corresponding to points projected on line $l$.
We compute a $(k, 4+\epsilon)$-VFTSWP $G_l$ for the set $S_l$ of points using the Algorithm~\ref{alg:addfts}.
For every edge $(r, s)$ in $G_l$, we introduce an edge $(p, q)$ in $G$, where $r$ (resp. $s$) is the geodesic projection of $p$ (resp. $q$) on $l$. 
Recursively, we apply the algorithm for the weighted points lying in the union of simple polygons corresponding to vertices of $G_d$ in $R'$, as well as for union of simple polygons corresponding to vertices of $G_d$ in $R''$. 
The recursion is continued till $R'$ (resp. $R''$) contains exactly one vertex of $G_d$.
Refer to Algorithm~\ref{alg:polydom}.

\begin{algorithm}[ht]
    \caption{VFTSWPPolygonalDomain(${\cal P}, S, k, \epsilon$)}
    \label{alg:polydom}

    \SetKwInOut{Input}{Input}
    \SetKwInOut{Output}{Output}

    \Input{A polygonal domain $\cal P$, a set $S$ of $n$ points located in the free space $\cal D$ of $\cal P$, a weight function $w$ that associates a non-negative weight to each point in $S$, an integer $k \geq 1$, and a real number $\epsilon > 0$.}
    \Output{A geodesic $(k,(12+\epsilon))$-VFTS for the set $S$ of weighted points.}

    Decompose $\mathcal{D}$ into $O(h)$ simple polygons such that each of these simple polygons has at most three splitting segments on its boundary. \\
    
    Compute the dual graph $G_d$ for the simple polygonal subdivision. \\

    Initialize the set $X$ to contain all the vertices of $G_d$. \\
    
    \While{$|X| \geq 1$}{
    
        \uIf{$|X| = 1$}{
            \scriptsize{} 
            Let $P'$ be the simple polygon corresponding to the vertex in $X$.
            Let $S'$ be the set of points in $S$ that are located in $P'$.  \\
            \normalsize{}

            VFTSWPSimplePolygon($P', S', k, \epsilon$). \\
        } \uElse {
        
            Using Theorem~\ref{thm4}, compute a $O(\sqrt{h})$-separator $R$ for the graph $G_d$.  
            Let $R', R'',$ and $R$ be the sets into which the vertices of $G_d$ is partitioned. \\
            
            \ForEach{splitting segment $l$ of the simple polygon that corresponds to every $r \in R$}{
                EdgesFromVFTSWPGeodProj($S, l, G$).
            }
        
        \scriptsize{} Let ${\cal P}'$ be the polygonal domain resulting from the union of simple polygons corresponding to each vertex in $R'$.  
        Also, let $S'$ be the set of points belonging to $S$ that are located in ${\cal P}'$. \normalsize{} \\
        VFTSWPPolygonalDomain(${\cal P}', S', k, \epsilon$).
    
        \scriptsize{} Let ${\cal P}''$ be the polygonal domain resulting from the union of simple polygons corresponding to each vertex in $R''$.  
        Also, let $S''$ be the set of points belonging to $S$ that are located in ${\cal P}''$. \normalsize{} \\
        VFTSWPPolygonalDomain(${\cal P}'', S'', k, \epsilon$).
        }
    } 
\end{algorithm}

\begin{algorithm}[ht]
    \caption{VFTSWPSimplePolygon($P, S, k, \epsilon$)}
    \label{alg:simppoly}

    \SetKwInOut{Input}{Input}
    \SetKwInOut{Output}{Output}
    
    \Input{A simple polygon $P$, a set $S$ of $n$ points located in $P$, a weight function $w$ that associates a non-negative weight to each point in $S$, an integer $k \geq 1$, and a real number $\epsilon > 0$.}
    \Output{A $(k,(12+\epsilon))$-VFTS $G$ for the weighted points in $S$.}

    \While{$|S| \geq 1$}{
       
        Using the algorithm in \cite{conf/jcdcg/BoseCKKM98}, compute a splitting segment $l$ for the set $S$ of points. \\

        EdgesFromVFTSWPGeodProj($S, l, G$). \\

        \scriptsize{} 
	Let $H_l$ be the left closed half-plane defined by $l$. 
        Let $P_l$ be the simple polygon $P \cap H_l$ 
        Also, let $S_l$ be set of points in $S$ belonging to $P_l$. 
	\normalsize{} \\
        VFTSWPSimplePolygon($P_l, S_l, k, \epsilon$). \\ 

        \scriptsize{} 
	Let $P_r$ be the simple polygon $P-P_l$.
        Also, let $S_r$ be the set of points in $S$ belonging to $P_r$. 
	\normalsize{} \\
        VFTSWPSimplePolygon($P_r, S_r, k, \epsilon$).
    }
\end{algorithm}

\begin{algorithm}[ht]
    \caption{EdgesFromVFTSWPGeodProj($S, l, G$)}

    \SetKwInOut{Input}{Input}
    \SetKwInOut{Output}{Output}

    \Input{A set $S$ of points, geodesic shortest path $l$, and a graph $G$.}
    \KwResult{ 
    Let $S_l = \cup_{p \in S} p_l$.
    Here, $p_l$ is a geodesic projection of $p$ on $l$.
    For each point $p_l \in S_l$, associate weight equal to $w(p) + d_{\pi}(p,p_l)$.
    Compute a spanner $G_l$ for the weighted points in $S_l$.
    Corresponding to every edge of $G_l$, introduce a new edge into $G$.}

    $S_l \leftarrow \phi$.
    
    \ForEach{$p \in S$}{
        Find a geodesic projection $p_l$ of $p$ on $l$. \\
            
        Assign a weight $w(p) + d_{\pi}(p,p_l)$ to $p_l$. \\
        
        $S_l := S_l \cup \{p_l\}$. \\
    }
        
    Using Algorithm~\ref{alg:addfts}, compute a $(k, 4+\epsilon)$-VFTS $G_l$ for the set $S_l$ of weighted points in $S$. \\
        
    For every edge $(p_l, q_l) \in G_l$, add an edge $(p, q)$ to $G$, where $p_l$ (resp. $q_l$) is a geodesic projection of $p$ (resp. $q$) on $l$.
\end{algorithm}

By following the algorithm in \cite{conf/compgeom/AbamAHA15}, for every simple polygon $P'$ that corresponds to a vertex of $G_d$ and the set $S'$ of points located in $P'$, we introduce more edges into $G$.
For every simple polygon $P''$ and the set $S''$ of points located in $P''$, a splitting segment that partitions $P''$ into two simple polygons such that each sub-polygon contains at most two-thirds of the points in $S''$ is called a {\it splitting segment with respect to $S''$ and $P''$}.
(In the following description, $S''$ and $P''$ are not mentioned with the splitting segment whenever they are clear from the context.)
Note that the splitting segment is defined in the context of both the simple polygon as well the polygonal domain.
Our algorithm partitions $P'$ into two simple sub-polygons $P_L'$ and $P_R'$ with a splitting segment $l$.
For every point $p \in S'$, we compute a geodesic projection $p_l$ of $p$ on $l$ and assign weight $w(p) + d_\pi(p, p_l)$ to point $p_l$.
Let $S_l'$ be the set comprising of all the geodesic projections of $S'$ on $l$.
Also, let $d_w$ be the weighted metric associated with points in $S_l$.
We use the algorithm from Section~\ref{sect:rd} to compute a $(k, 4+\epsilon)$-VFTSWP $G_l$ for the metric space $(S_l, d_w)$.
For every edge $(r, s)$ in $G_l$, we add an edge between $p$ and $q$ to $G$ with weight $d_\pi(p, q)$, wherein $r$ (resp. $s$) is the geodesic projection of $p$ (resp. $q$) on $l$.
Let $S_L'$ (resp. $S_R'$) be the set of points in sub-polygon $P_L'$ (resp. $P_R'$) of $P'$.
We recursively process $P_L'$ (resp. $P_R'$) with points in $S_L'$ (resp. $S_R'$) unless $|S_L'|$ (resp. $|S_R'|$) is zero. 
Refer to Algorithm~\ref{alg:simppoly}.

We prove the graph $G$ resulted from including all the edges as described above is a $(k, 12+\epsilon)$-vertex fault-tolerant spanner for weighted points in $S$.
Later, we modify the above algorithm so that it results in $4+\epsilon$ approximation.
We show after removing any subset $S'$ with $|S'| \leq k$ from $G$, for any two points $p$ and $q$ in $S \setminus S'$, there exists a path between $p$ and $q$ in $G \setminus S'$ such that the $d_G(p, q)$ is at most $(12+15\epsilon)d_w(p, q)$.
First, note that there exists a splitting segment $l$ at some iteration of the algorithm so that $p$ and $q$ are on different sides of $l$. 
Let $r$ be a point belonging to $l \cap \pi(p, q)$.
Let $S'_l$ be the set comprising of geodesic projections of points in $S'$ on $l$.
Since $G_l$ is a $(k,(4+5\epsilon))$-VFTS for the metric space $(S_l, d_w)$, there exists a path $\tau$ between $p_l$ and $q_l$ in $G_l \setminus S'_l$ whose length is upper bounded by $(4+5\epsilon) \cdot d_w(p_l, q_l)$. 
Let $\tau'$ be a path between $p$ and $q$ in $G \setminus S'$, which is obtained by replacing each vertex $v_l$ of $\tau$ by $v$ in $S$ such that the point $v_l$ is the geodesic projection of $v$ on $l$.
In the following lemma, we show the length of $\tau'$, which is $d_{G \setminus S'}(p, q)$, is upper bounded by $(12+15\epsilon) \cdot d_w(p, q)$.
In proving this, we exploit  (i) the triangle inequality of geodesic shortest paths, (ii) specific weight associated to each point resulted from a geodesic projection, (iii) non-negative weights of points in $S$, and (iv) the spanning ratio as well as the vertex fault-tolerance provided by the graph $G_l$.

\begin{lemma}
\label{lem20}
Given a set $S$ of $n$ points located in the free space $\cal D$ of a polygonal domain $\cal P$, a weight function $w$ to associate a non-negative weight to each point in $S$, a positive integer $k$, and a real number $\epsilon > 0$, Algorithm~\ref{alg:polydom} computes a $(k, 12+\epsilon)$-vertex fault tolerant geodesic spanner for the weighted points in $S$.
\end{lemma}
\begin{proof}
Using induction on the number of points, we show that there exists a $(12 + 15\epsilon)$-spanner path between $p$ and $q$ in $G \setminus S'$ for any set $S'$ of vertices of $G$ with $|S'| \le k$.
The induction hypothesis assumes $d_{G \setminus S'}(p, q) \le (12+15\epsilon) \cdot d_w(p, q)$ for any set $S'$ of vertices of $G$ with $|S'| \le k$ and for any two points $p, q$ in $G \setminus S'$.
Consider a set $S' \subset S$ such that $|S'| \leq k$.
Let $p, q$ be two arbitrary points in $S \setminus S'$.
For any shortest path $\pi(p, q)$ between $p$ and $q$, there exists a splitting segment $l$ such that $\pi(p, q)$ intersects $l$.
Let $r$ be a point of intersection of $\pi(p, q)$ and $l$.
Also, let $p_l$ (resp. $q_l$) be a geodesic projection of $p$ (resp. $q$) on $l$.
Since $G_l$ is a $(k, 4+5\epsilon)$-VFTSWP, there exists a path $\tau$ between $p_l$ and $q_l$ in $G_l$ with length at most $(4+5\epsilon) d_w(p_l, q_l)$.
By replacing each vertex $x_l$ of $\tau$ by $x \in S$ such that $x_l$ is the projection of $x$ on $l$, gives a path between $p$ and $q$ in $G \setminus S'$. 
Thus, the length of the path $d_{G \setminus S'}(p,q)$ is less than or equal to the length of path $\tau$ in $G_l$.
For every $x,y \in S \setminus S'$,
{\setlength{\abovedisplayskip}{0pt}
\begin{flalign}
\hspace{6mm}d_{G \setminus S'}(p,q) &\le \sum_{x_l,y_l \in \tau} d_w(x,y)&&\nonumber\\
    &= \sum_{x_l,y_l \in \tau} (w(x) + d_{\pi}(x,y) + w(y))&&\nonumber\\
    &\leq \sum_{x_l,y_l \in \tau} (w(x) + d_{\pi}(x,x_l) + d_{\pi}(x_l,y_l) + d_{\pi}(y_l,y) + w(y))&&\nonumber\\
        &\text{[by triangle inequality of geodesic shortest paths]}&&\nonumber\\
    &= \sum_{x_l,y_l \in \tau} (w(x_l) + d_{\pi}(x_l,y_l) + w(y_l))&&\nonumber\\
        &\text{[since the weight associated with projection} \ z_l \ \text{of every point} \ z \ \text{is} \ w(z) + d_{\pi}(z,z_l)]&&\nonumber\\
    	&\leq (4+5\epsilon) \cdot d_w(p_l,q_l)&&\label{simppoly:lweq} \\
\hspace{20mm}&\leq \sum_{x_l,y_l \in \tau} d_w(x_l,y_l)&&\nonumber\\
        &\text{[since} \ G_l \ \text{is a} \ (k,(4+5\epsilon))\text{-vertex fault-tolerant geodesic spanner]}&&\nonumber\\
        &= (4+5\epsilon) \cdot [w(p_l) + d(p_l,q_l) + w(q_l)]&&\nonumber\\
        &= (4+5\epsilon) \cdot [w(p_l) + d_{\pi}(p_l,q_l) + w(q_l)]&&\nonumber\\
    &\text{[since points $p_l$ and $q_l$ are located on $l$]}&&\nonumber\\
        &= (4+5\epsilon) \cdot [w(p) + d_{\pi}(p,p_l) + d_{\pi}(p_l,q_l) + d_{\pi}(q_l,q) + w(q)]\label{eq20}&&\nonumber\\
    &\text{[since the weight associated with } \ z_l \ \text{of point} \ z \ \text{is} \ w(z) + d_{\pi}(z,z_l) \text{]}&&\nonumber\\
    &\leq (4+5\epsilon) \cdot [w(p) + d_{\pi}(p,r) + d_{\pi}(p_l,q_l) + d_{\pi}(r,q) + w(q)]&&\nonumber\\
    &\text{[since $r$ is a point belonging to both $l$ as well as $\pi(p, q)$]}&&\nonumber\\
	&= (4+5\epsilon) \cdot [w(p) + d_{\pi}(p, q) + w(q) + d_{\pi}(p_l,q_l)]&&\nonumber\\
    &= (4+5\epsilon) \cdot [d_w(p, q) + d_{\pi}(p_l,q_l)].&&\nonumber
\end{flalign}}

\noindent
But,
{\setlength{\abovedisplayskip}{0pt}
\begin{flalign}
d_{\pi}(p_l,q_l) &\leq d_{\pi}(p_l,p) + d_{\pi}(p,q) + d_{\pi}(q,q_l)&&\nonumber\\
	&\text{[by triangle inequality of geodesic shortest paths]}&&\nonumber\\
        &\leq d_{\pi}(p,r) + d_{\pi}(p,q) + d_{\pi}(r,q)&&\nonumber
\end{flalign}}

{\setlength{\abovedisplayskip}{0pt}
\begin{flalign}
\hspace{12mm}
	&\text{[since $r$ belongs to both $l$ and $\pi(p,q)$]}&&\nonumber\\
        &\leq w(p) + d_{\pi}(p,r) + w(p) + d_{\pi}(p,q) + w(q) + d_{\pi}(r,q) + w(q)&&\nonumber\\
        &\text{[since weight associated with every point is non-negative]}&&\nonumber\\
    &\le 2d_w(p,q)&&\nonumber\\
    &\text{[since $r$ belongs to $\pi(p, q)$].}&\nonumber
\end{flalign}}

Hence, $d_{G \setminus S'}(p,q) \leq 3(4+5\epsilon) \cdot d_w(p,q)$.
\end{proof}

We further improve the stretch factor of $G$ by applying the refinement given in \cite{conf/soda/AbamBS17} to the above described algorithm.
In doing this, for each point $p \in S$, we compute the geodesic projection $p_{\gamma}$ of $p$ on the splitting segment $\gamma$, and construct a set $S(p,\gamma)$ as defined herewith. 
Let $\gamma(p) \subseteq \gamma$ be $\{ x \in \gamma : d_w(p_{\gamma},x) \leq (1 + 2\epsilon) d_{\pi}(p,p_{\gamma}) \}$. 
We divide $\gamma(p)$ into $c$ pieces with $c \in O(\frac{1}{\epsilon^{2}})$, and these pieces are denoted by $\gamma_{j}(p)$, for $1 \leq j \leq c$.
For each piece $j$, we compute the point $p_\gamma^{(j)}$ nearest to $p$ in $\gamma_j(p)$.
The set $S(p,\gamma)$ is defined as $\{ p_\gamma^{(j)} : p_\gamma^{(j)} \in \gamma_{j}(p)$  and $1 \leq j \leq c\}$. 
For every $r \in S(p,\gamma)$, the non-negative weight $w(r)$ of $r$ is set to $w(p) + d_{\pi}(p,r)$.
Let $S_\gamma$ be $\cup_{p \in S}S(p,\gamma)$.

We replace the set $S_l$ in Algorithm~\ref{alg:polydom} with the set $S_\gamma$, and compute a $(k,(4+5\epsilon))$-VFTSWP $G_l$ using the algorithm from Section~\ref{sect:rd} for the set $S_\gamma$.
Further, for every edge $(r,s)$ in $G_l$, we include an edge $(p,q)$ of weight $d_\pi(p, q)$ into $G$ whenever $r \in S(p,l)$ and $s \in S(q,l)$.
The rest of the algorithm is same as Algorithm~\ref{alg:polydom}.
In the following, we restate a lemma from \cite{conf/soda/AbamBS17}, which is useful for our analysis.

\begin{lemma}
\label{lemfromabam}
Consider any two points $x, y \in {\cal P}$ .
Let $r$ be the point at which a shortest path between $x$ and $y$ intersects a splitting segment $\gamma$.
Also, let $x_\gamma, y_\gamma$ respectively be the geodesic projections of $x$ and $y$ on $\gamma$.
Then, $d_{\pi}(x, x_\gamma) + d_\gamma(x_\gamma, r) \le (1 + \epsilon) \cdot d_{\pi}(x,r)$ and
$d_{\pi}(y, y_\gamma) + d_\gamma(y_\gamma, r) \le (1 + \epsilon) \cdot d_{\pi}(y,r)$.
\end{lemma}

\begin{theorem}
\label{thm:polydom}
Given a set $S$ of $n$ points located in the free space $\cal D$ of a polygonal domain $\cal P$ with $h$ holes, a weight function $w$ to associate a non-negative weight to each point in $S$, a positive integer $k$, and a real number $\epsilon > 0$, there exists a $(k, 4+\epsilon)$-vertex fault tolerant geodesic spanner with $O(\frac{k n \sqrt{h+1}}{\epsilon^{2}}\lg{n})$ edges for weighted points in $S$.
\end{theorem}

\begin{proof}
Let $P'$ be a simple polygon resultant from the partitioning of the free space of $\cal P$ with the splitting segments. 
Also, let $f(n')$ be the number of edges added to $G$ due to a set of $n'$ points located in $P'$.
Then, $f(n') = f(n_l') + f(n_r') + (\frac{k n'}{\epsilon^{2}})$, where $n_l', n_r'$ are the number of points in each of the partitions formed by the splitting segment.
Since both $n_l' \ge \frac{n}{3}$ and $n_r' \geq \frac{n}{3}$, $f(n')$ is $O(\frac{k n'}{\epsilon^{2}} \lg n)$.
The number of edges included in $G$ due to all the $O(h)$ simple polygons together is $O(\frac{k n \sqrt{h+1}}{\epsilon^{2}} \lg n)$.
Analogously, for the polygonal domain $\cal P$, at each recursive level except for the last, our algorithm is including $O(\frac{k n''}{\epsilon^{2}} \lg n'')$ edges into $G$. 
Here, $n''$ is the number of points at that recursive level.
After including the edges added at the last level of the recursion tree, the total number of edges in $G$ is $O(\frac{k n \sqrt{h+1}}{\epsilon^{2}} \lg{n})$.
(The one added to $h$ in the time complexity considers the case of a polygonal domain with zero holes, i.e., a simple polygon.)

Next, we prove the stretch factor of graph $G$.
Consider any set $S' \subset S$ such that $|S'| \leq k$ and two arbitrary points $p$ and $q$ from the set $S \setminus S'$. 
We show that there exists a $(4 + 14\epsilon)$-spanner path between $p$ and $q$ in $G \setminus S'$.
Let $l$ be the splitting segment that intersects a shortest path between $p$ and $q$.
If $r \notin l$, then $p_l'$ (resp. $q_l'$) is equal to $p_l$ (resp. $q_l$). 
Otherwise, $p_l'$ (resp. $q_l'$) is a point in $S(p,l)$ (resp. $S(q,l)$) that is nearest to $p$ (resp. $q$).
Then,

{\setlength{\abovedisplayskip}{0pt}
\begin{flalign}
	d_w(p_l',q_l') &= w(p_l') + d(p_l',q_l') + w(q_l')&&\nonumber\\
        &\leq w(p_l') + d(p_l',r) + d(r,q_l') + w(q_l')&&\nonumber\\
        &\text{[by triangle inequality of shortest paths]}&\nonumber\\
	&\leq w(p_l') + d(p_l',r) + w(r) + w(r) +  d(r,q_l') + w(q_l')&&\nonumber
\end{flalign}}

{\setlength{\abovedisplayskip}{0pt}
\begin{flalign}
\hspace{12mm}
        &\text{[since the weight associated with each point is non-negative]}&&\nonumber\\
	&= w(p) + d_{\pi}(p,p_l') + d(p_l',r) + w(r) + w(r) + d(r,q_l') + d_{\pi}(q_l',q) + w(q)&&\nonumber\\
        &\text{[due to weight assigned to geodesic projection of points]}&&\nonumber\\
	&\leq w(p) + (1 + \epsilon) \cdot d_\pi(p,r) + w(r) + w(r) + (1+\epsilon) d_\pi(r,q) + w(q)&&\nonumber\\
	&\text{[due to Lemma~\ref{lemfromabam}]}&&\nonumber\\
	&\leq (1+\epsilon)(d_w(p,r) + d_w(r,q))&&\nonumber\\
	&= (1+\epsilon)d_w(p,q)&&\nonumber\\
	&[\text{since }r \in l \cap \pi(p,q)].&\nonumber
\end{flalign}}
\noindent
Substituting above in inequality (\ref{simppoly:lweq}), $d_{G \setminus S'}(p,q) \leq (4+5\epsilon) (1+\epsilon) \cdot d_w({p_l}',{q_l}') \leq (4+5\epsilon) (1+\epsilon)\cdot d_w(p,q) \le (4+14\epsilon) \cdot d_w(p,q)$.
\end{proof}

\section{Vertex fault-tolerant spanner for the weighted points located on a polyhedral terrain}
\label{sect:terrains}

Given a polyhedral terrain $\cal T$, a set $S$ of $n$ points located on $\cal T$, a weight function $w$ to associate a non-negative weight to each point in $S$, a positive integer $k$, and a real number $\epsilon > 0$, in this section we describe an algorithm to compute a geodesic $(k, 4+\epsilon)$-vertex fault-tolerant spanner with $O(\frac{kn\lg{n}}{\epsilon^2})$ edges for the set $S$ of weighted points.

We denote the boundary of $\mathcal{T}$ with $\partial \mathcal{T}$.
Also, we denote a geodesic Euclidean shortest path (a path lying on $\cal T$) between any two points $a, b \in \mathcal{T}$ with $\pi(a, b)$.
The distance along $\pi(a, b)$ is denoted by $d_\pi(a, b)$. 
For any two points $x, y \in \partial \mathcal{T}$, we denote the closed region lying to the right (resp. left) of $\pi(x,y)$ when going from $x$ to $y$, including (resp. excluding) the points lying on the shortest path $\pi(x,y)$ with $\pi^{+}(x,y)$ (resp. $\pi^{-}(x,y)$).
The {\it projection} $p_\pi$ of a point $p \in {\cal T}$ on a shortest path $\pi$ between two points lying on $\mathcal{T}$ is defined as a point on $\pi$ that is at the minimum geodesic distance from $p$ among all the points of $\pi$.
For three points $u,v,w \in \mathcal{T}$, the closed region bounded by shortest paths $\pi(u,v)$, $\pi(v,w)$, and $\pi(w,u)$ is called an {\it sp-triangle}, denoted by $\Delta(u, v, w)$.
If points $u, v, w \in \mathcal{T}$ are clear from the context, we denote the sp-triangle with $\Delta$.
In the following, we restate a theorem from \cite{conf/soda/AbamBS17}, which is useful for our analysis.

\begin{theorem}
\label{thm7}
For any set $P$ of $n$ points on a polyhedral terrain $\mathcal{T}$, there exists a balanced sp-separator: a shortest path $\pi(u,v)$ connecting two points $u,v \in \partial \mathcal{T}$ such that $\frac{2n}{9} \leq |\pi^{+}(u,v) \cap P| \leq \frac{2n}{3}$, or a sp-triangle $\Delta$ such that $\frac{2n}{9} \leq |\Delta \cap P| \leq \frac{2n}{3}$.
\end{theorem}

Thus, an sp-separator is either bounded by a shortest path (in the former case), or by three shortest paths (in the latter case).

First, a balanced sp-separator as given in Theorem~\ref{thm7} is computed. 
Let $\gamma$ be a shortest path that belongs to an sp-separator.
The sets $S_{in}$ and $S_{out}$ comprising of points are defined as follows: if the sp-separator is a shortest path, then define $S_{in}$ to be $\gamma^{+}(u,v) \cap S$; otherwise, $S_{in}$ is $\Delta \cap S$; and, the set $S_{out} = S - S_{in}$.
For each point $p \in S$, we compute a projection $p_{\gamma}$ of $p$ on every shortest path $\gamma$ of sp-separator, and associate a weight $w(p) + d_{\mathcal{T}}(p, p_{\gamma})$ with $p_{\gamma}$.
Let $S_{\gamma}$ be the set $\cup_{p \in S}\hspace{0.02in} p_{\gamma}$.
Using Algorithm~\ref{alg:addfts} from Section~\ref{sect:rd}, we compute a $(4+\epsilon)$-spanner $G_{\gamma}$ for the set $S_{\gamma}$ of points located on $\gamma$.
Further, for each edge $(p_{\gamma}, q_{\gamma})$ in $G_{\gamma}$, an edge $(p,q)$ is included in $G$, where $p_{\gamma}$ (resp. $q_{\gamma}$) is the projection of $p$ (resp. $q$) on $\gamma$. 
More edges are added to $G$ while recursively processing points in sets $S_{in}$ and $S_{out}$.
Refer to Algorithm~\ref{alg:terr}.

\begin{algorithm}[ht]
    \caption{VFTSWPTerrain(${\cal T}, S, k, \epsilon$)}
    \label{alg:terr}

    \SetKwInOut{Input}{Input}
    \SetKwInOut{Output}{Output}
    
    \Input{A polyhedral terrain $\cal T$, a set $S$ on $n$ points located in $\cal T$, a weight function $w$ that associates a non-negative weight to each point in $S$, an integer $k \geq 1$, and a real number $0 < \epsilon \le 1$.}
    \Output{A $(k,(12+\epsilon))$-VFTS $G$.}
    
    \While{$|\mathcal{T} \cap S| \geq 1$}{
        
        Compute a balanced sp-separator $\Gamma$ for $\mathcal{T}$ using the algorithm given in \cite{conf/soda/AbamBS17}. \\
        
        \ForEach{bounding shortest path $\gamma$ of $\Gamma$}{
            EdgesFromVFTSWPGeodProj($S, \gamma, G$). \\
        }

        \scriptsize{}
        Let ${\cal T}'$ be $\pi^{+}(u,v)$ if the balanced sp-separator is a shortest path $\pi(u,v)$; otherwise, let ${\cal T}'$ be $\Delta$.
        Also, let $S_{in}$ be the set of points located on ${\cal T}'$. \\
        \normalsize{}
        VFTSWPTerrain(${\cal T}', S_{in}, k, \epsilon$). \\
    
        \scriptsize{}
        Let ${\cal T}''$ be $\pi^{-}(u,v)$ if the balanced sp-separator is a shortest path $\pi(u,v)$; otherwise, let ${\cal T}''$ be ${\cal T} \setminus \Delta$.
        Also, let $S_{out}$ be the set of points located on ${\cal T}''$. \\
        \normalsize{}
        VFTSWPTerrain(${\cal T}'', S_{out}, k, \epsilon$). 
    }
        
\end{algorithm}

The rest of the algorithm in constructing $G$ is same as the algorithm given in Section~\ref{sect:polydom}.
To prove graph $G$ is a geodesic $(k,(12+15\epsilon))$-VFTSWP for the weighted points in $S$, we induct on the number of points. 
Consider any set $S' \subset S$ such that $|S'| \leq k$, and any two arbitrary points $p$ and $q$ from the set $S \setminus S'$.
The induction hypothesis assumes when the number of points is less than $n$, the graph $G$ is a $(k, 12+15\epsilon)$-VFTSWP.
As part of the inductive step, we consider the set $S$, having $n$ points.
For the case of both $p, q \in S_{in}$ or both of them belong to $S_{out}$, by induction hypothesis (as the number of points in $S_{in}$ or $S_{out}$ is less than $n$), there exists a $(12+15\epsilon)$-spanner path between $p$ and $q$ in $G \setminus S'$. 
Otherwise, if $p \in S_{in}$ and $q \in S_{out}$ or, $q \in S_{in}$ and $p \in S_{out}$.
Without loss of generality, assume the former holds.
Let $r$ be a point on $\gamma$ at which the geodesic shortest path $\pi(p,q)$ intersects $\gamma$.
Also, let $p_\gamma, q_\gamma$ respectively be the geodesic projections of $p$ and $q$ on $\gamma$.  
Since $G_\gamma$ is a $(k, 4+5\epsilon)$-VFTS, there exists a path $\tau$ between $p_\gamma$ and $q_\gamma$ in $G_\gamma$ with length at most $(4+5\epsilon).d_w(p_\gamma, q_\gamma)$. 
Let $\tau'$ be the path obtained by replacing each vertex $x_\gamma$ of $\tau$ by $x \in S$, for $x_\gamma$ being the geodesic projection of $x$ on $\gamma$. 
Note that $\tau'$ is a path between $p$ and $q$ in $G \setminus S'$.
The rest of the proof in showing that there exists a path between $p$ and $q$ in $G \setminus S'$ such that $d_G(p, q)$ is at most $(12+15\epsilon)d_w(p,q)$ is same as in the proof of Lemma~\ref{lem20}.

To obtain a $(4+\epsilon)$-VFTSWP, same as in Section~\ref{sect:polydom}, we apply the refinements to each shortest path in all the sp-separators.
The following theorem is an immediate consequence of these refinements. 
The proof of this theorem is the same as the proof of Theorem~\ref{thm:polydom} for points in a polygonal domain.

\begin{theorem}
\label{lem26}
Given a set $S$ of $n$ points located on a polyhedral terrain $\cal T$, a weight function $w$ to associate a non-negative weight to each point in $S$, a positive integer $k$, and a real number $\epsilon > 0$, there exists a $(k, 4+\epsilon)$-vertex fault tolerant geodesic spanner with $O(\frac{k n}{\epsilon^{2}}\lg{n})$ edges for the weighted points in $S$.
\end{theorem}

\section{Conclusions}
\label{sect:conclu}

In this paper, we presented algorithms to achieve $k$ vertex fault-tolerance when the points are associated with non-negative weights.
We devised algorithms to compute a $(k, 4+\epsilon)$-VFTSWP when the input points belong to any of the following: $\mathbb{R}^{d}$, simple polygon, polygonal domain, and on the terrain.
It would be interesting to achieve a better bound on the stretch factor when each point is associated with either a unit or a zero weight.
Apart from efficient computation, finding the worst-case lower bound on the number of edges for the fault-tolerant spanners for weighted points could be interesting. 
Besides, the future work in the context of these kinds of spanners could include finding a relation between vertex fault-tolerance and edge fault -tolerance, and optimizing maximum degree, diameter, and weight.

\subsection*{Acknowledgement}

This research of R. Inkulu is supported in part by NBHM grant 248(17)2014-R\&D-II/1049, and SERB MATRICS grant MTR/2017/000474.

\bibliographystyle{plain}


\end{document}